\documentclass[11pt]{article}

\usepackage[dvipdfmx]{graphicx}
\usepackage{color}
\usepackage{mathtools}
\usepackage{todonotes}
\usepackage{amsmath,amsfonts,amssymb,amsthm}
\usepackage{ascmac}
\usepackage{lineno}
\usepackage{amsfonts}
\usepackage[normalem]{ulem}
 \usepackage{fancybox}
 \usepackage{algorithm,algorithmic}
\usepackage{url}
\usepackage{colortbl}
\usepackage{tabularx}
\usepackage{subcaption}
\usepackage{cleveref}
\usepackage{multirow, hhline}
\usepackage{comment}

\newtheorem{lemma}{Lemma}[section]
\newtheorem{theorem}{Theorem}[section]

\newtheorem{corollary}{Corollary}[section]

\usepackage{lscape}
\usepackage{a4wide}

\usepackage[dvipdfmx]{graphicx}

\title{\textsc{Turning Tiles} is PSPACE-complete}
\author{Kanae Yoshiwatari, Hironori Kiya, Koki Suetsugu, Tesshu Hanaka, Hirotaka Ono}

\date{}

\begin{document}
\maketitle
\begin{abstract}
In combinatorial game theory, the winning player for a position in normal play is analyzed and characterized via algebraic operations. Such analyses define a value for each position, called a game value. A game (ruleset) is called universal if any game value is achievable in some position in a play of the game. Although the universality of a game implies that the ruleset is rich enough (i.e., sufficiently complex), it does not immediately 
imply that the game is intractable in the sense of computational complexity. 
This paper proves that the universal game \textsc{Turning Tiles} is PSPACE-complete.
\end{abstract}

\section{Introduction}

A combinatorial game is a game played by two players and there is no randomness and both players know all the information about the game.
Any position of combinatorial games can be associated with a value, called game value.
Roughly speaking, if positions $P_1$ and $P_2$ have the same game value, the game trees rooted at $P_1$ and $P_2$ have the same canonical structure. 

A ruleset, we always say as the name of a game, is a set of rules of a game.
Some rulesets include any game value and are called \emph{universal}. That is, under a universal ruleset, any game value is achievable in some position in a play. 
The notion of a universal ruleset is defined in \cite{carvalho2019nontrivial}, and the same paper proves that \textsc{Generalized Konane} is a universal ruleset.
The ruleset that was secondarily shown to be universal is \textsc{Turning Tiles} and since then, two other rulesets were also shown to be universal by (polynomial-time) reductions from \textsc{Turning Tiles}
~\cite{suetsugu2022discovering,suetsugu2023new}.
The universality of a game implies that it is sufficiently complex. Then, a natural question arises: 
how is the relation between the universality of a game and its computational complexity?  
Actually, \textsc{Generalized Konane} was shown to be PSPACE-complete in~\cite{hearn2005amazons} before the proof of \cite{carvalho2019nontrivial}.  
This paper shows that \textsc{Turning Tiles} is also PSPACE-complete. Since the two other universal rulesets are polynomially reducible from \textsc{Turning Tiles}, they are also PSPACE-complete. Consequently, the rulesets currently known to be universal are all PSPACE-complete.   

\medskip 


\textsc{Turning Tiles} is a board game played by blue and red players using two-sided square tiles and one token. One side of a tile is the face colored blue or red, and the other side is the back colored black.
The board forms a rectangle and is filled with square tiles, where some are faced up (blue or red) and others are faced down (black).   
The token is initially placed on a black tile. The board has a coordinate system, which has directions of north, south, east, and west. 
Two players alternately move the token straightly on faced-up tiles with their own colors. 
When a player moves the token to a new tile with their color, the player turns the tile into black, so the game is named \textsc{Turning Tiles}.  
A player can continue to move the token as long as tiles of the player's color continue in the same direction (e.g., east) and can stop it at an arbitrary tile. This is one player's turn, and the next player starts to move the token on the black tile.  
A player who cannot move on their turn loses.
The above rule on \textsc{Turning Tiles} uses only one token, but we can extend the rule to use multiple tokens in a straightforward way. 
\textsc{Turning Tiles} is a game, but we also use it as the problem name for deciding the winner of a given position. Other game names also follow this style.

\section{Preliminaries}
We assume the basic knowledge of graph theory and combinatorial graph theory. 

We show that \textsc{Turning Tiles} is PSPACE-hard by 
a reduction from the \textsc{Generalized Geography}.
\textsc{Generalized Geography} is a two-player game played on a directed graph. 
The instance of \textsc{Generalized Geography} consists of $G=(V,A)$ and a token on a vertex $s$, called the start vertex.
If $G$ is bipartite, we may use $G=(U,V,A)$ instead of $G=(V,A)$.  

In each turn, a player moves the token to a vertex $v$ adjacent to the vertex with the token such that the token has never visited $v$. 
The player who cannot move in their turn loses.

\begin{lemma}[\cite{lichtenstein1980go}]
\textsc{Generalized Geography} on a bipartite planar graph whose maximum degree is 3 is PSPACE-complete, where the start vertex $s$ has indegree 0 and outdegree 2.
\end{lemma}

Here, the maximum degree of a directed graph $G$ means the maximum degree of the underlying graph of $G$. 

Note that \textsc{Generalized Geography} itself is an impartial game, but in \textsc{Generalized Geography} on bipartite $G=(U,V,A)$ one player can choose a vertex only in $U$ and the other can choose a vertex only in $V$; we can consider that one player is the blue player and the other is red. 

\section{PSPACE-completeness of {\normalfont\textsc{Turning Tiles}}}

\begin{theorem}
\textsc{Turning Tiles} using one token is PSPACE-complete.
\end{theorem}

\begin{proof}

We assume that \textsc{Turning Tiles} is played on an $N\times N$ board, which defines the input size. 
The number of options in each position is at most $4N$, and one game ends after at most $N^2$, because the number of blue or red tiles is monotonically decreasing. This implies that determining the winner of \textsc{Turning Tiles} is in PSPACE. 

We then show PSPACE-hardness by a reduction from 
\textsc{Generalized Geography} as mentioned above. 
Without loss of generality, we assume that in \textsc{Turning Tiles} the first player is the blue player,  and in \textsc{Generalized Geography} on a bipartite graph $G=(V_B,V_R,A)$ the start vertex is in $V_B$.
Note that the first player moves a token from $v \in V_B$ to $v' \in V_R$ and the second player moves a token from $v' \in V_R$ to $v \in V_B$.
For a given position (instance) $G$ of \textsc{Generalized Geography}, we construct a position $T$ of \textsc{Turning Tiles} such that the first blue player wins in $T$ if and only if the first player wins in $G$. 

Before explaining the details, we give an outline of our reduction. For a vertex and an edge in \textsc{Generalized Geography}, we prepare a vertex gadget and a connection gadget, respectively. Vertex gadgets are linked by connection gadgets in reflecting $G$ of \textsc{Generalized Geography}. Each vertex gadget forms a rectangle associated with a color, i.e., blue or red. A blue (resp., red) vertex gadget corresponds to a blue (resp., red) vertex in \textsc{Generalized Geography}. The rectangle of a blue (resp., red) vertex gadget is surrounded by black tiles except for at most three blue (resp., red) tiles, called \emph{portals}. Portals have two types: in-portal and out-portal. In the descriptions below, we use symbols $s_i$ and $t_i$ for in-portal and out-portal, respectively, where the subscript $i$ represents the name of the corresponding in-neighbor or out-neighbor.     
In connection gadgets, tiles with blue or red colors induce a path, and each connection gadget links an out-portal of a blue (resp., red) vertex gadget and an in-portal of a red (resp., blue) vertex gadget. In the following, we present in detail vertex and connection gadgets in \textsc{Turning Tiles}.

\medskip 

\paragraph{Gadget for the start vertex $s$}
Let $s$ be the start vertex of $G$. Recall that $s$ is a blue vertex with only two outgoing edges, say $(s,v_1)$ and $(s,v_2)$. 
For $s$, we design the corresponding gadget in $T$ as Figure \ref{start}. 
In this gadget, what the first blue player can do is to move the token to $a$, and what the second red player can do is to move the token to $b$.
In the next turn, the blue player has two options: 
move the token to $c$ or $d$. 
If the blue player stops the token at $c$, the token eventually reaches out-portal $t_1$, and otherwise (i.e., the token is put at $d$), the token eventually reaches out-portal $t_2$.
%
We can see that the series of these actions correspond to the actions of the first player of \textsc{Generalized Geography}: going to $v_1$ and going to $v_2$. 


\begin{figure}[htbp]
 \centering
 \includegraphics[width=0.5\linewidth]{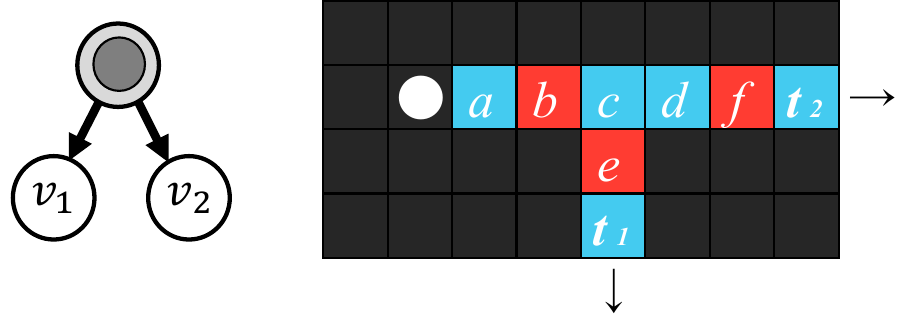}
 \caption{The start vertex}
 \label{start}
\end{figure}

Next, we design blue vertex gadgets of the other vertices in $V_B$. 
The red vertex gadgets for $V_R$ are designed by exchanging the color of tiles and the roles of players. 
%
%
Recall that the maximum degree of $G$ is $3$. It follows that the outdegree of each vertex $v$ in $G$ can be only $0$, $1$, or $2$, which corresponds to the number of options of the player at $v$.  
Note that $G$ does not have a vertex with indegree $0$ except for the start vertex because it is never played in a game. 
Also, without loss of generality, we can assume that a vertex $v$ with outdegree $0$ has indegree $1$, because if $v$ has two incoming edges $(u_i,v)$ and $(u_j,v)$, they can be replaced with $(u_i,v)$ and $(u_j,v')$, which does not change the winner. 
In the following, we explain how we construct gadgets for vertices by outdegree. 



\paragraph{Gadgets for vertice with outdegree 0}

We first see the gadgets of $v\in V_B$ with outdegree $0$ and indegree $1$ by the above assumption. 
In \textsc{Generalized Geography}, the blue player in $v$ cannot move, i.e., loses, 
and thus we design the gadgets so as to reflect the property.     
The gadget forms a pair of blue and red tiles as shown in Figure \ref{0_123}. 
When the blue player puts the token on $s_i$ from the outside of the gadget, the red player puts the token on the red tile next to $s_i$. Then no place is left to put the token for the blue player, and the blue player loses the game.



\begin{figure}[htbp]
  \begin{minipage}[b]{0.48\columnwidth}
    \centering
 \includegraphics[width=\columnwidth]{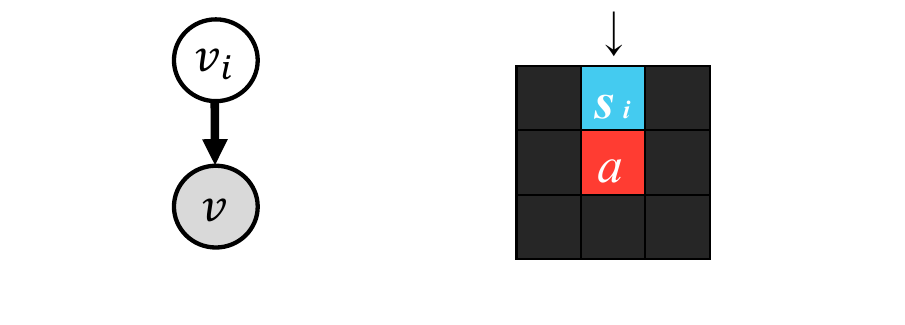}
\caption{Gadgets for outdegree-$0$ vertices in $V_B$}
\label{0_123}
  \end{minipage}
  \hspace{0.04\columnwidth} 
  \begin{minipage}[b]{0.48\columnwidth}
    \centering
     \includegraphics[width=\columnwidth]{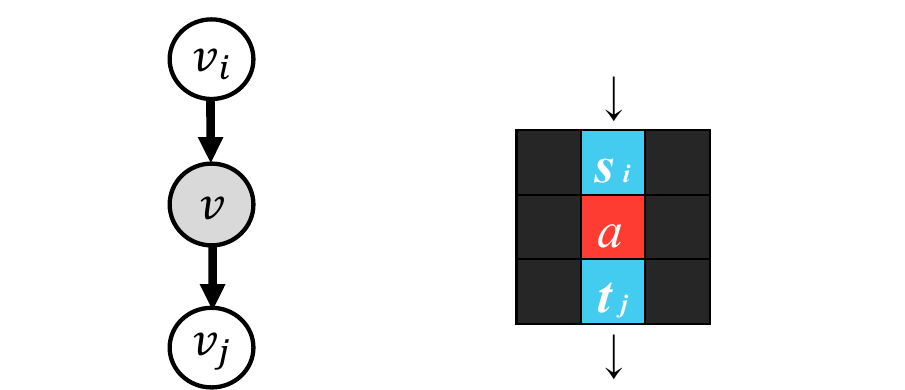}
     \caption{a vertex $v \in V_B$ which is outdegree $1$ and indegree $1$}
    \label{1_1}
  \end{minipage}
\end{figure}


\paragraph{Gadgets for vertice with outdegree 1}

We next see the gadgets of $v\in V_B$ with outdegree $1$. 
The indegree of $v$ is $1$ or $2$. 
For $v$ with indegree $1$, the blue player of \textsc{Generalized Geography} has no option, i.e., just move the token to the neighboring vertex. 
%
Figure \ref{1_1} shows the gadget of $v$ with indegree $1$, where the blue player passes the turn by moving the token twice. 



For $v$ with indegree $2$, Figure \ref{1_2} shows the gadget, where $s_i$ and $s_j$ are the incoming tiles. 
We consider the case when the blue player puts the token on $s_1$ from the outside of the gadget.  
Then, the options of the red player are moving the token to $a$, $b$, or $c$. 
Among these, $a$ and $c$ are not good options because 
the red player loses by the blue player's moving the token to $d$ and $e$, respectively. 
Thus, the red player chooses to move the token on $b$.
%
In the next turn, what the blue player can do is to put the token at $t_k$; the token leaves the gadget. 
By symmetry, also when the blue player puts the token on $s_j$ from the outside of the gadget, the token eventually reaches $t_k$. From these arguments, this gadget can simulate the role of $v$ with indegree $2$. 
\begin{figure}[htbp]
 \centering
 \includegraphics[width=0.5\linewidth]{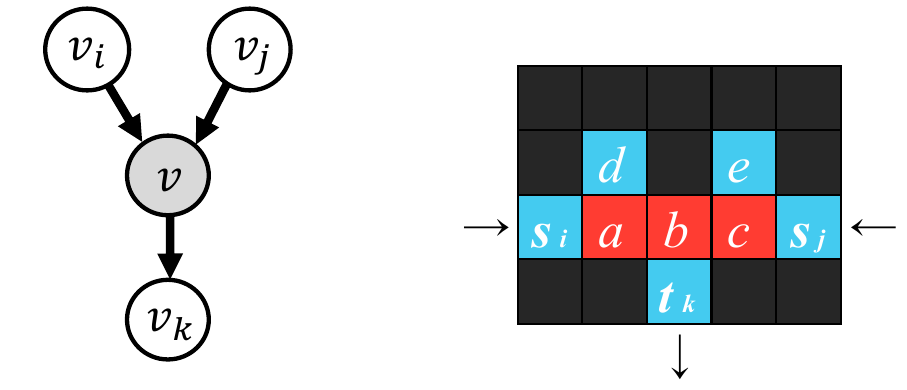}
 \caption{A vertex $v \in V_B$ with outdegree $1$ and indegree $2$}
 \label{1_2}
\end{figure}

Note that the token may visit this gadget twice since it has two in-portals. Such a situation corresponds to a forbidden move of the red player in \textsc{Generalized Geography}. This means that if the red player leads the token to this gadget twice, then the player must lose. We can confirm that this actually holds as Figure \ref{1_2_after}. 


\begin{figure}[tbp]
 \centering
 \includegraphics[width=0.5\linewidth]{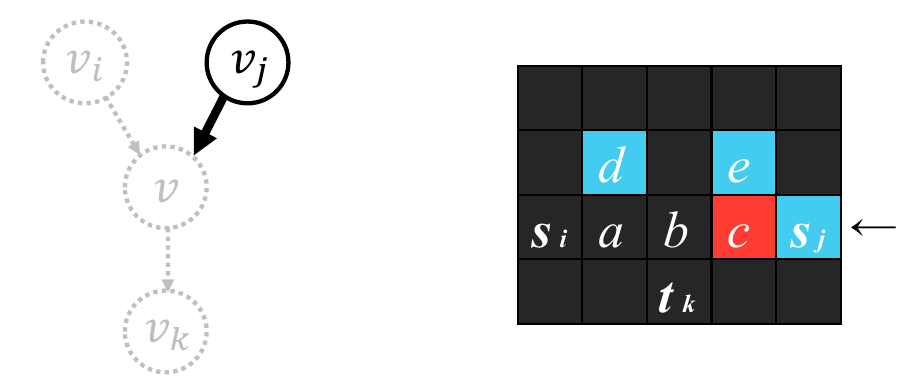}
 \caption{The position after the token passes through the gadget of $v$ in Figre \ref{1_2}}
 \label{1_2_after}
\end{figure}


\paragraph{Gadgets for vertice with outdegree 2}
Next, we consider a vertex $v \in V_B$ whose outdegree is $2$ and indegree is $1$.
As mentioned above, \textsc{Generalized Geography}, the outdegree 2 of $v\in V_B$ represents two options of the blue player. The corresponding gadget in \textsc{Turning Tiles} is shown in Figure \ref{2_1}, which has a similar structure to the gadget for the starting vertex $s$. 
By a similar argument to the gadget for $s$, the blue player has two options: moving the token to $b$ or $c$, 
which eventually leads it to $t_1$ or $t_2$, respectively. 
   

\begin{figure}[tbp]
 \centering
 \includegraphics[width=0.5\linewidth]{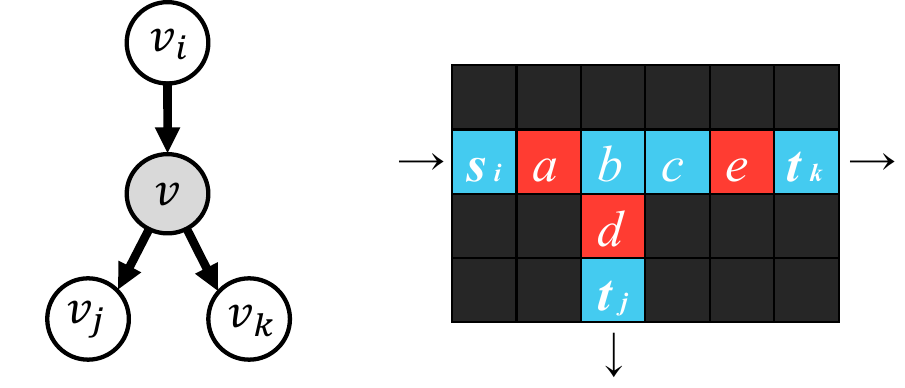}
 \caption{a vertex $v \in V_B$ which is outdegree $2$ and indegree $1$}
 \label{2_1}
\end{figure}


\paragraph{Connection gadgets}
Finally, we describe how to arrange gadgets on a planar grid board.
For any $v_b \in V_B$ and any $v_r \in V_R$, we connect $t_j$ of the gadget of $v_b$ and the $s_i$ of the gadget of $v_r$ if $(v_b,v_r) \in A$.
By using the gadgets of Figure \ref{path} with a polynomially sufficiently large space, we can connect pairs of $s_i$ and $t_j$ without crossing, since a planar graph $G$ can be embedded in a grid graph with an $O(|V|)$ area~\cite{valiant81}.
Also, if the parities of the gadgets are different, we use the gadget in Figure \ref{parity}.
After connecting all gadgets, all remaining squares are filled with black tiles.

\begin{figure}[bp]
  \begin{minipage}[b]{0.5\columnwidth}
    \centering
    \includegraphics[width=.8\columnwidth]{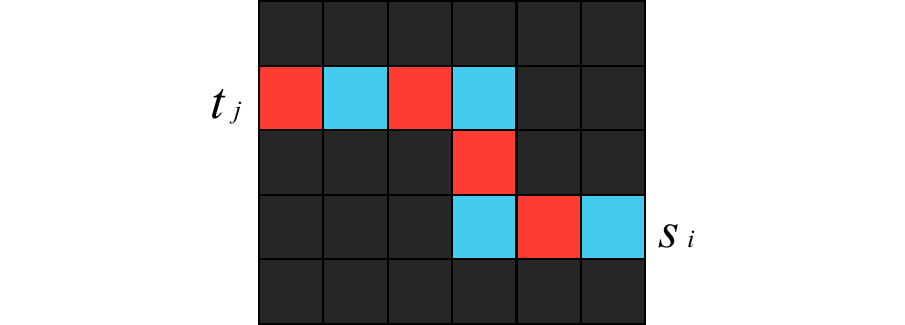}
    \caption{A connection gadget}
    \label{path}
  \end{minipage}
  \hspace{0.04\columnwidth} 
  \begin{minipage}[b]{0.5\columnwidth}
    \centering
    \includegraphics[width=.8\columnwidth]{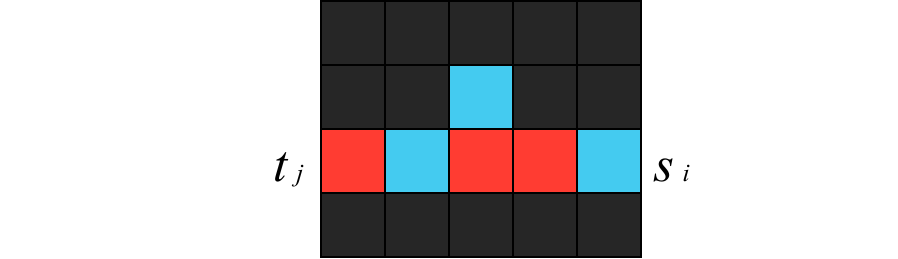}
  \caption{A parity gadget}
  \label{parity}
  \end{minipage}
\end{figure}





Let us see an example of the reduction.
From a position of \textsc{Generalized Geography} in \Cref{graph_example}, we 
construct 
the corresponding position of \textsc{Turning Tiles} in  \Cref{tile_example}, where both the winners are the same.
The operations of the reduction can be done in polynomial time of the input size.
\end{proof}
%
%

\begin{figure}[htbp]
\centering
\begin{subfigure}{0.25\textwidth}
    \includegraphics[width=\textwidth]{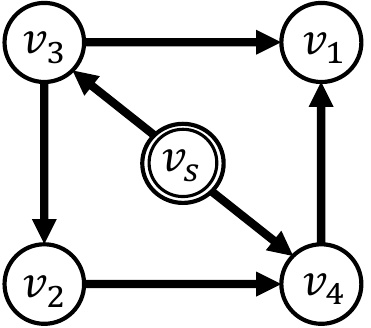}
 \caption{An instance of \textsc{Generalized Geography} on a bipartite and planar graph}
 \label{graph_example}
\end{subfigure}
\hfill
\begin{subfigure}{0.45\textwidth}
    \includegraphics[width=\textwidth]{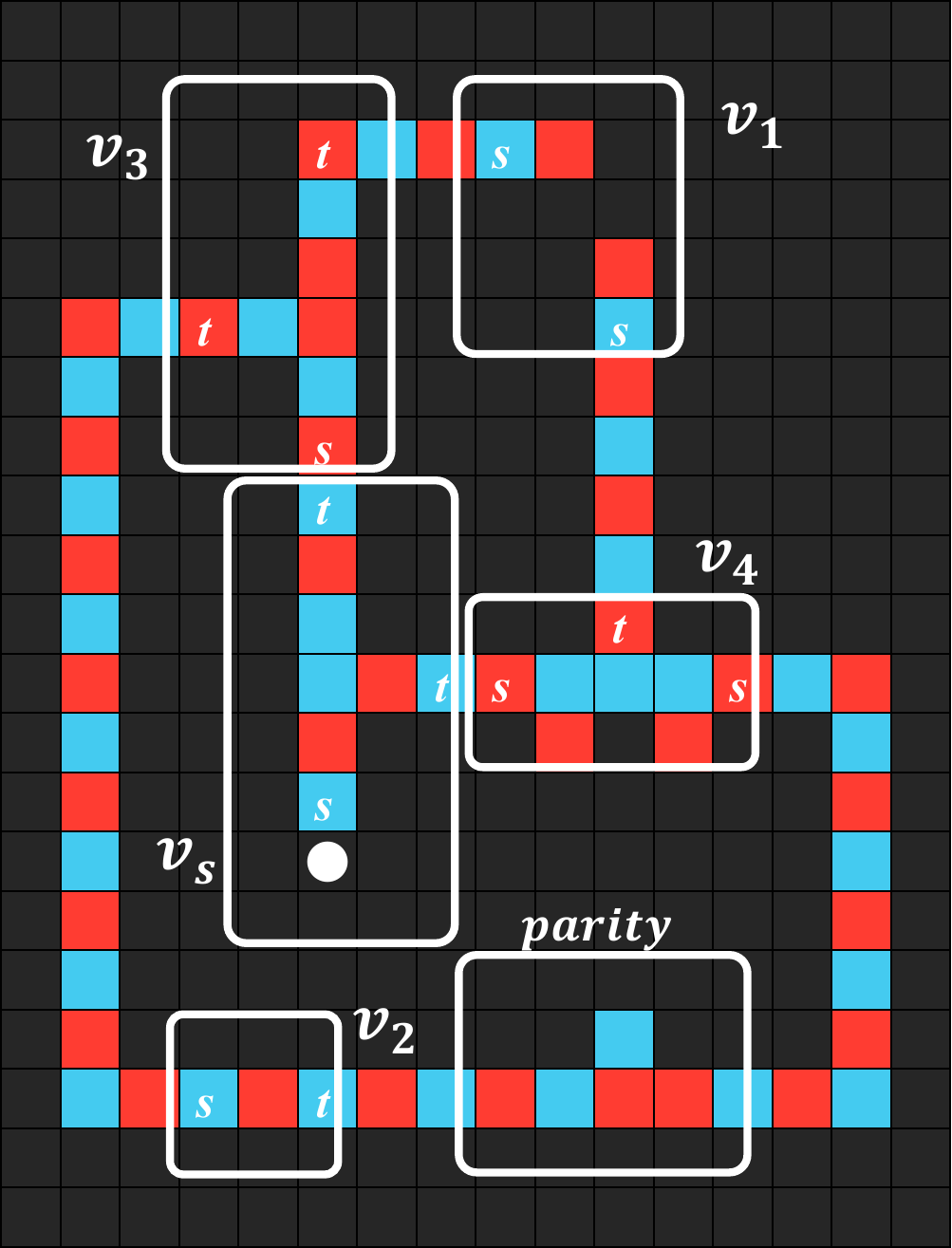}
 \caption{A position of \textsc{Turning Tiles} constructed from the instance of \Cref{graph_example}}
 \label{tile_example}
\end{subfigure}        
\caption{A full example of the reduction}
\label{fig:reduction_example}
\end{figure}
%

We can easily extend the result to \textsc{Turning Tiles} with multiple tokens because the proof works even if we add isolated tokens to some unused areas in the reduction.

\begin{corollary}
\textsc{Turning Tiles} using any number of tokens is PSPACE-complete.
\end{corollary}

Another corollary is about other games (rulesets). Suetsugu \cite{suetsugu2023new} introduces two other rulesets,  \textsc{Go on Lattice} and \textsc{Beyond the door}, which are also shown to be universal rulesets. The proof of their universality is shown by polynomial-time reductions, which imply that their winner decisions are also PSPACE-hard. Since it is easy to see that they belong to PSPACE, we have the following corollary.

\begin{corollary}
\textsc{Go on Lattice} is PSPACE-complete.     
So is \textsc{Beyond the door}. 
\end{corollary}


\end{document}